\documentclass[12pt]{article}
\pdfoutput=1
\usepackage{graphicx}
\usepackage{mathrsfs}
\usepackage{amssymb}
\usepackage{amsmath}
\usepackage{graphicx}
\usepackage{amsmath,amsfonts}
\usepackage{algorithm}
\usepackage{algorithmicx}
\usepackage{algpseudocode}
\usepackage{amscd}
\usepackage[mathscr]{eucal}
\usepackage{lineno}
\usepackage{extarrows}

\usepackage{listings}
\usepackage{tikz}

\newtheorem{theorem}{Theorem}[section]
\newtheorem{thm}{Theorem}[section]
\numberwithin{thm}{section}
\newtheorem{remark}[thm]{Remark}
\newtheorem{definition}[thm]{Definition}
\newtheorem{lemma}[thm]{Lemma}

\newenvironment{proof}{\noindent\\ \noindent\relax{\sc
     Proof}}{{\samepage\par\nopagebreak\hbox
     to\hsize{\hfill$\Box$}}}
\newcommand{\be}{\begin{equation}} \newcommand{\ee}{\end{equation}}
\newcommand{\bd}{\begin{displaymath}} \newcommand{\ed}{\end{displaymath}}
\newcommand{\ba}{\begin{align}} \newcommand{\ea}{\end{align}}
\newcommand{\baa}{\begin{align*}} \newcommand{\eaa}{\end{align*}}
\newcommand{\ben}{\begin{enumerate}} \newcommand{\een}{\end{enumerate}}
\newcommand{\bi}{\begin{itemize}} \newcommand{\ei}{\end{itemize}}

\newcommand{\ud}{\mathrm{d}}
\newcommand{\E}[1]{\operatorname{E}\left[ #1 \right]}

\algnewcommand\algAnd{\textbf{and}~}
\algnewcommand\algOr{\textbf{or}~}

\usepackage[normalem]{ulem}

\algnewcommand\And{\textbf{and}}

\begin{document}


\title{}
\author{Krzysztof Bartoszek} 

\maketitle

\begin{abstract}
We consider a hybrid method to simulate the return time to the initial state
in a critical--case birth--death process. The expected value of this return time
is infinite, but its distribution asymptotically follows a power--law. Hence,
the simulation approach is to directly simulate the process, unless the simulated time exceeds
some threshold and if it does, draw the return time from the tail of the power law.
\end{abstract}

Keywords : 
birth--death process, infinite mean, phylogenetic tree, power--law distribution, return time

\section{Introduction: a model for phylogenetic trees}\label{secIntro}
Birth--and--death processes are frequently used today to model various branching 
phenomena, e.g. phylogenetic trees. Empirically observed (or rather estimated from e.g.
genetic data) phylogenies can exhibit multiple patterns, e.g. many co--occurring 
species or only a  dominating one at a given time instance. The HIV
phylogeny is an example of the former, while the influenza phylogeny 
of the latter. During a given season there is one main flu virus going around,
but it may change between seasons. 

In \cite{TLigRSch2009}  a very similar model that (depending on the choice
of a parameter $\lambda$) can generate both patterns was proposed. They model the amount
of types alive at a given time instance as follows. Assume that at time $t$ there are
$N(t)>0$ types present. Then, at time $t$ the birth rate of the types is $\lambda N(t)$ and
the death rate is $N(t)$. Each type has at birth, independently of the other types, a
fitness value attached to it. If a death event occurs, the type with the lowest fitness goes extinct.
The type with the largest value of the fitness is called the dominating type.
As only the ranking of the types' fitness is relevant the distribution of 
the fitness values does not matter. 

The  main result of \cite{TLigRSch2009}
is the characterization of the lifetime of the dominating type.
\begin{theorem}[Thm. $1$ in \cite{TLigRSch2009}]
Take $\delta \in (0,1)$. If $\lambda \le 1$, then
$$
\lim\limits_{t\to \infty} P(dominating~types~at~\delta t~and~t~are~the~same)=\delta,
$$
while if $\lambda>1$, then this limit is $0$.\label{thm1}
\end{theorem}
Obviously, if $\lambda \le 1$, then a given (maximal) type can persist (like influenza) 
but when $\lambda>1$, then there will be frequent switching between dominating types
(like HIV---we cannot observe which strain dominates). 
The key object in the proof of Thm. \ref{thm1} are the random times of hitting state $1$, conditional on
having started in state $2$, we denote this random variable by $H$. Then, the waiting time for
returning from state $1$ to state $1$ can be represented as $X+H$, where $X\sim\exp(\lambda)$.
Naturally state $0$ is an absorbing state of the process. However, we do not worry
about this, as it is assumed that if there is only one type, then it cannot die \cite{TLigRSch2009},
i.e. $0$ will never be reached and the process starts anew in $1$.

Let $F(t)$ denote the cumulative distribution function (cdf) of $H$, i.e.
$F(t) = P(H\le t)$.
Here we will focus on the critical case $(\lambda=1)$. In \cite{TLigRSch2009} it was
claimed that in this regime $F(t)=t/(t+1)$ (proof of Lemma $3$ therein). However,
in \cite{KBarMKrz2014} it was shown that this statement is not true (cf. Eqs. $3.3$ and $3.4$ therein) 
however the right asymptotic behaviour of $F(t)$ is known
(cf Eq. $3.2$ therein), 
\be\label{eqFtAsympt}
\lim\limits_{t\to \infty} t(1-F(t)) = 1.
\ee

To be able to work with the law of $H$ we introduce the following definition.
\begin{definition}\label{dfPowerLaw}
We say that a random variable $Y$ follows a power--law probability distribution with parameter $\gamma> 0$, 
if it has support on $(y_{\min},\infty)$, $y_{\min}\ge 0$, and
its law asymptotically satisfies
$$
P(Y>y) \sim C y^{-\gamma },~~~~~ \mathrm{i.e.}~ \lim\limits_{y\to \infty} y^{\gamma }P(Y>y)  = C,
$$
for some constant $C$.
\end{definition}
Moments of order $m\ge \gamma$ are infinite. 

\begin{lemma}
$H$ is a positive random variable that follows a power law distribution with $C=1$, $\gamma=1$ and
$$\E{H}=\infty.$$
\end{lemma}
\begin{proof}
$H$ is positive by construction as a random sum of exponential random variables.
From Eq. \eqref{eqFtAsympt} we know that
$\lim\limits_{y\to \infty} yP(H>y) = 1$, corresponding to a power law distribution with 
$C=1$, $\gamma=1$ in Defn. \ref{dfPowerLaw}.
Hence, for $m\ge 1$ it holds directly that 
$$\E{H^{m}} = \int\limits_{0}^{\infty}P(H^{m}>y)\ud y= \infty.$$
\end{proof}

\begin{remark}
It is worth noting that 
the infinite mean can be derived directly from the model formulation also. Denote by $H_{j}$ the time to reach state $1$
from state $j$ and $h_{j}=\E{H_{j}}$. Then, $H=H_{2}$ and define $H_{1}=0$. Notice that by the memoryless property
of the process 
we can see that $h_{j}$ is a non--decreasing sequence, to get from state $j$ to state $1$ one
has to return to state $j-1$. 

By model construction we have $h_{2}=1/4+h_{3}/2$ and 
$$
h_{3} = \frac{1}{6} + \frac{1}{2}h_{2}+\frac{1}{2}h_{4},
$$
giving
$$
h_{2} = h_{3} +(h_{3}-h_{4}) - \frac{1}{3}.
$$
In the same way one will have for all $j \ge 3$
$$
h_{j}-h_{j+1} = h_{j+1}+(h_{j+1}-h_{j+2}) - \frac{1}{j+1}
$$
resulting in for every $N\ge 3$
$$
h_{2} = h_{3} +(h_{N}-h_{N+1}) -\sum\limits_{i=3}^{N} \frac{1}{i}.
$$
As obviously $h_{2}\ge 0$,
one needs for every $N \ge 3$
$$
h_{3} \ge h_{N+1}-h_{N} + \sum\limits_{i=3}^{N} \frac{1}{i} \ge \sum\limits_{i=3}^{N} \frac{1}{i}
$$
implying that $h_{3} = \infty$. As $h_{2}=1/4+h_{3}/2$, then immediately $h_{2}=\infty$.
\end{remark}

\section{Simulation algorithm}\label{secSimAlg}
Very often an important component of studying stochastic models is the possibility to simulate them.
This is in order to illustrate the model, gather intuition for its properties, back--up 
(i.e. check for errors) analytical results and to design Monte Carlo based estimation or testing
procedures. A Markov process as described in Section \ref{secIntro} is in principle trivial to simulate
using the Gillespie algorithm \cite{DGil1977}.
Let $N_{n}$ be the embedded Markov Chain, i.e. $N_{n}=N(t_{n})$, where $t_{n}$ is the time of the
$n$--th birth or death event.
If the process is in state $N_{n}=i$ at step $n$, then one draws an exponential with rate
$(i+i\lambda)$ waiting time and $N_{n+1}=i+\xi$, where $P(\xi=1)=\lambda i/(i+\lambda i)$ and $P(\xi=-1)=i/(i+\lambda i)$.
However, in the critical case $\lambda=1$, as we showed that $\E{H}=\infty$, we cannot expect to reliably
sample the full distribution of the chain's trajectory. The tails will be significantly undersampled. 
In the simplest case if we want to plot an estimate of $H$'s density (e.g. a histogram), then
we can expect its right tail to be significantly (in the colloquial sense) too light. 

Unfortunately, only the asymptotic behaviour of $H$'s survival function, $1-F(t)$, is known. Therefore,
if we were to draw from this form, we should firstly expect the initial part of the histogram to be 
badly found. 
Furthermore, as the constant $C=1$ in Eq. \ref{eqFtAsympt}, then we know that the cumulative distribution function
related to $F(t)$, $1-1/t$, has non--zero support 
on $[1,\infty)$. 
Obviously $H$ can take values also in $(0,1)$, so the above does not suffice.

Therefore, in this work we propose a hybrid algorithm that combines both approaches. 
Intuitively the left side of the histogram is simulated directly, while the right side is drawn
from the 
asymptotic. 
We have as our aim a proof of concept study---to see if reasonable results
can be obtained, leaving improvements of the algorithm and its analytical properties for 
further work.

Let $p$ be the proportion of simulations that we do not want to simulate directly but draw from the asymptotic.  
Define $T_{\min}$ as the value for which
$$p = P(H > T_{\min}) = 1-F(T_{\min}) \approx T_{\min}^{-1}.$$ This gives $T_{\min}=p^{-1}$.
We will use $p^{-1}$ and $T_{\min}$ interchangeably depending on the focus---the
tail probability or the threshold value.

\begin{algorithm}[!htp]
\caption{Simulating return time to state $1$}\label{algSimH}
\begin{algorithmic}[1]
\State Initialize $p$, $T_{\min}=p^{-1}$, $N=2$, $H=0$
\While {$H \le T_{\min}$ \algAnd $N\neq 1$}
    \State draw $u \sim \exp(2N)$
    \State $H=H+u$
    \If {
    $H \le T_{\min}$}
    \State $N=N+\xi$ \Comment{
    $P(\xi=1)=P(\xi=-1)=\frac{1}{2}$}
    \EndIf
\EndWhile
\If {$H > T_{\min}$ \algOr $N\neq 1$}
    \State $H=$\lstinline{poweRlaw::rplcon(1,xmin=}$T_{\min}$\lstinline{,alpha=2)} \Comment{See Section \ref{secSimAlg}}
\EndIf
\State draw $t\sim \exp(1)$
\State \Return{$(H,H+t)$} \Comment{$H+t$ is return time from $2$ to $1$}
\end{algorithmic}
\end{algorithm}

The hybrid simulation approach is described in Alg. \ref{algSimH}. Until the waiting time 
does not exceed a certain threshold the simulations proceeds directly according to the
model's description. The threshold simulation is chosen so that (approximately, according to 
the limit distribution) with probability $1-p$ it will not be exceeded. If the threshold 
is exceeded, then $H$ is drawn from a law corresponding to the survival function's asymptotic
behaviour. There is a trade--off in the choice of $p$. If we choose a small $p$, i.e. a threshold
far in the right tail, then with high probability we will have an exact simulation algorithm.
However, on the other hand the running time may be large. In contrast, taking a larger $p$ will
reduce the running time, however, more samples will not be drawn exactly but from the
asymptotic, and our final simulated value's distribution could be further
away from its true law. Obviously, $p\in(0,1)$ by it being a probability and the previously discussed
properties of the asymptotics of the survival function, i.e. $C=1$.

In order to draw from the law corresponding to the asymptotic behaviour
of the right tail of the survival function, we use the powerRlaw \cite{CGil2015} \texttt{R} \cite{R}
package. The \lstinline{poweRlaw::rplcon(1,}$T_{\min}$\lstinline{,}$\alpha$\lstinline{)} draws a single value
from a power law supported on $(T_{\min},\infty)$ with density 
and cumulative distribution functions (Eq. $(1)$, $(4)$, \cite{CGil2015})  equalling 
\be\label{eqpowerRlawDens}
p(t) = \frac{\alpha-1}{T_{\min}} \left(\frac{t}{T_{\min}} \right)^{-\alpha},~~~P(T\le t)=1-\left(\frac{t}{T_{\min}}\right)^{-\alpha+1}
\ee
for $\alpha>1$ and $T_{\min}>0$. Note the correspondance between Defn. \ref{dfPowerLaw} and Eq. \eqref{eqpowerRlawDens},
$\gamma=\alpha-1$.

In our situation we have $1-F(t)\sim t^{-1}$. Hence, we can take the power law corresponding
to the limit as the 
cdf $F_{\infty}(t)=(1-t^{-1})\mathbf{1}_{t>1}$,
with density equalling $f_{\infty}(t)=t^{-2}$ on $(1,\infty)$. This implies taking $\alpha=2$
when calling \lstinline{poweRlaw::rplcon()}. As we have defined a threshold of $p^{-1}$,
we need to include it also when drawing the value, i.e. we do not want to have empty draws of too small
values. Setting, then $T_{\min}=p^{-1}$, tells \lstinline{poweRlaw::rplcon()} that our law
is concentrated on $(T_{\min},\infty)$. Notice that we are working with the conditional 
random variable $H$, given that $H > T_{\min}$. Its conditional cdf 
will equal $p^{-1}(F_{\infty}(t)-F_{\infty}(T_{\min}))\mathbf{1}_{t>T_{\min}}$, with associated density 
$p^{-1}t^{-2}\mathbf{1}_{t>T_{\min}}$, but this equals $p(t)=(1/p)^{-1}(xp^{-1})^{-2}\mathbf{1}_{t>T_{\min}}$.
The function \lstinline{poweRlaw::rplcon()} draws a value using the inverse 
cdf method.
Let $U\sim \mathrm{Unif}[0,1]$, and then $T$ from \lstinline{poweRlaw::rplcon()} is given by
$$
T=T_{\min}(1-U)^{-1/(\alpha-1)}.
$$

As we do not know from what value the tail asymptotics are a good approximation we 
cannot a priori be sure whether the threshold will be exceeded with probability $p$, nor
if after $p^{-1}$ the sampling of $H$ from the limit will be accurate. These important
properties will be checked empirically in the simulation study in Section \ref{secSimulStudy}.

\section{Simulation setup and results}\label{secSimulStudy}
The simulation study presented here has a number of aims, to illustrate the distribution of waiting times $H$
as simulated exactly and via Alg. \ref{algSimH}. Secondly, to study whether Alg. \ref{algSimH} is a sensible
approach to to simulating this random time. 
All simulations were done in R version $3.4.2$ \cite{R} running on an openSUSE $42.3$ ($x86$\_$64$) box
with a $3.50$GHz Intel\textsuperscript{\textregistered} Xeon\textsuperscript{\textregistered} CPU.

Simulating the Markov process directly (and then extracting $H$)
is a straightforward procedure. However, as $\E{H}=\infty$, we introduce a cutoff, if the number
of steps exceeds a given number, here $10^{8}$, we end the simulation and mark that the process
did not return to state $1$. Out of the $10^{7}$ repeats, $767$ reached the maximum allowed number of steps, $10^{8}$.
In order to see how well Alg. \ref{algSimH} is corresponding to the 
true distribution of $H$, we re--run it for two values of $p$, $p_{1}=0.0001$ and $p_{2}=0.0005$.
Then, we compare the logarithms of the survival functions, of both simulations and furthermore
on the interval $(p_{2}^{-1}, p_{1}^{-1})$, Fig. \ref{figAlgECDFs}. 
The first sample, with threshold
$p_{1}^{-1}$, can be thought of as the ``true one'' on this interval, as $H$ was simulated exactly---directly from
the model's definition. We present the logarithm of the survival function as otherwise nothing
would be visible from the plot, due to the heavy tail. The power--law property of the tail
can be clearly seen in the left panel of Fig. \ref{figAlgECDFs}. In fact, if one regress $\log(y)$
on $\log(x)$ one will obtain a slope estimate of $-1.010$ $(p_{1})$ or $-0.9375$ $(p_{2})$.
This is in agreement with our result that $1-F(t) \sim t^{-1}$.
Out of the $10^{7}$ simulations only $992$ had $H>p_{1}^{-1}=10000$
in the case of the $p_{1}$ simulation and $4992$ instances 
had $H>p_{2}^{-1}=2000$ in the case of the $p_{2}$ simulation. We can see that these correspond nearly
exactly to the desired proportion of exceeding the cutoff, i.e. 
$992/10^{7}=0.0000992\approx p_{1}=0.0001$ and
$4992/10^{7}=0.0004992\approx p_{2}=0.0005$.

\section{Discussion}\label{secDisc}
The results presented here indicate that a hybrid approach for simulating the waiting
time to return to state $1$ in the considered birth--death model is a promising one. 
From Figs. \ref{figNumSteps} and \ref{figFinalStates} it is evident that the
direct approach cannot capture the heavy tail. When, the simulation
was stopped due to exceeding a time/step limit these figures show that the process
was usually far on the right and not approaching $1$. 

For our considered hybrid procedure to be effective one needs to appropriately
choose the threshold $p^{-1}$. If it is too small, then the return time might
not yet be in the asymptotic regime. If it is too large, then the running time 
can be too long. In our case $10^{7}$ repeats with $p_{2}=0.0005$ took 
about $68.5$ hours while for $p_{1}=0.0001$ it took about $15.6$ days. 
Definitely $p_{1}^{-1}$ is too large for a threshold, while $p_{2}^{-1}$ will
be acceptable given that the required sample size is smaller. 
Most importantly,
the comparison between $p_{1}$'s and $p_{2}$'s results, in Fig. 
\ref{figAlgECDFs}, showed that 
with $p_{2}=0.0005$ the hybrid approach yields samples that do not seem
to be distinguishable from the true law. This is as on the interval
$(p_{2}^{-1},p_{1}^{-1})$ the $p_{1}$ simulation is exact. 

Our study here points to three possible, exciting future directions of
development. Firstly, to identify an optimal value for $p$. Secondly,
to develop a better simulation approach so that when $H$ exceeds $p$
and the draw ``has to be made from the tail'', the simulated path does
not need to be discarded, i.e. to characterize the law of the return time 
from an arbitrary state $m$ to state $1$. And lastly the study of 
$1-F(t)-t^{-1}$ i.e. how does the survival function deviate from its 
asymptotic behaviour. 

\begin{figure}[!ht]
\centering
\includegraphics[width=0.31\textwidth]{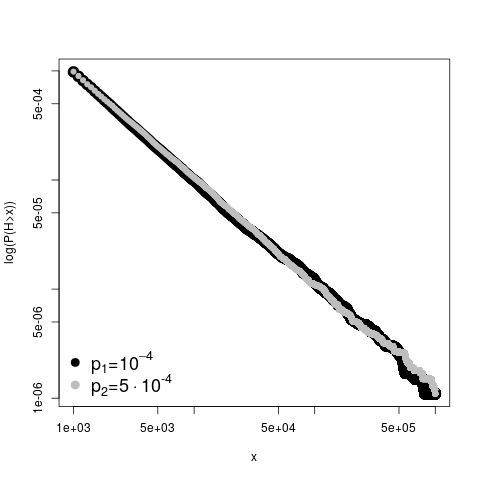}  
\includegraphics[width=0.31\textwidth]{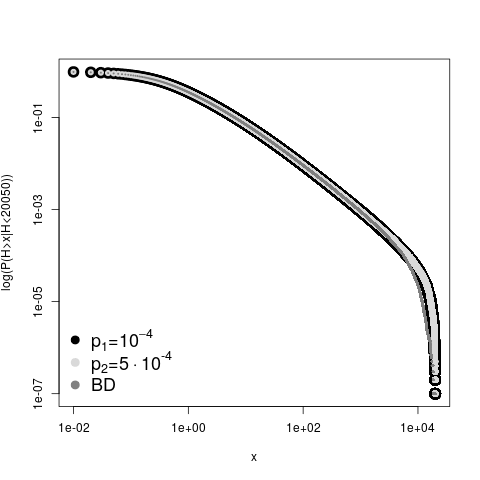}  
\includegraphics[width=0.31\textwidth]{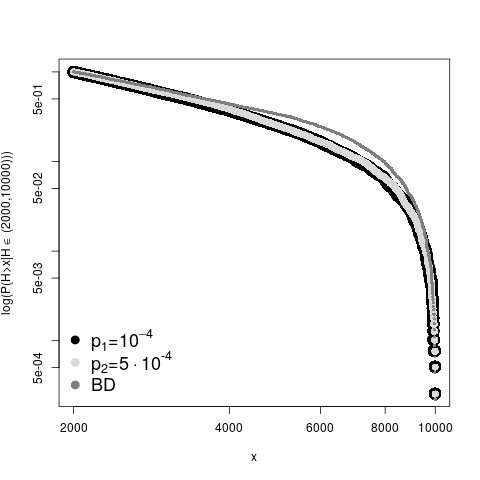}  
\caption{
Logarithms of survival functions created using \texttt{R}'s \lstinline{ecdf()}
function to estimate the cdf.
Left: comparison of the survival functions for $p_{1}$ and $p_{2}$,
centre: comparison of the survival functions for $p_{1}$, $p_{2}$ and 
the ``direct'' simulation (on the interval $(0,20050)$, where the ``direct'' simulation
produced values), right: comparison of the survival functions for $p_{1}$, $p_{2}$
and the ``direct'' simulation on the interval $(p_{2}^{-1}, p_{1}^{-1})$.
For readability the right plot starts from $1000$.
In the centre plot the ``direct'' simulation is conditional
on the number of steps being lesser than $10^{8}$, so this
survival function is a conditional one.
The Kolmogorov--Smirnov test, \lstinline{ks.test()} function, was 
used to check for differences between the empirical
distributions. All p--values exceeded $0.277$, except
for the (right plot) comparison on the interval $(p_{2}^{-1}, p_{1}^{-1})$,
between ``direct'' and $p_{1}$ simulation $(5.551\cdot 10^{-16})$
and ``direct'' and $p_{2}$ simulation $(2.331 \cdot 10^{-15})$.
Here, the p--value when testing the $p_{1}$ versus the $p_{2}$ simulation
equalled $0.277$.
The $x$--axes are on the log scale.
}\label{figAlgECDFs}
\end{figure}

\begin{figure}[!ht]
\centering
\includegraphics[width=0.35\textwidth]{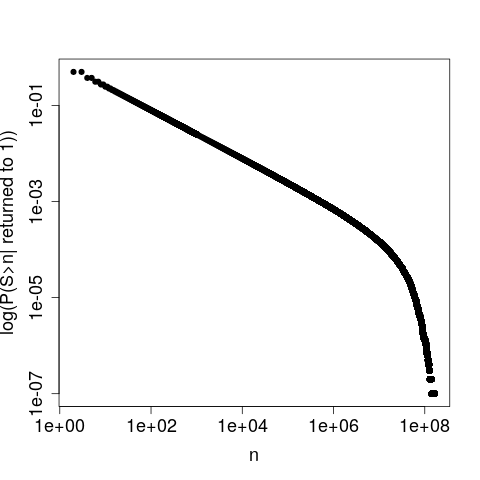}  
\includegraphics[width=0.35\textwidth]{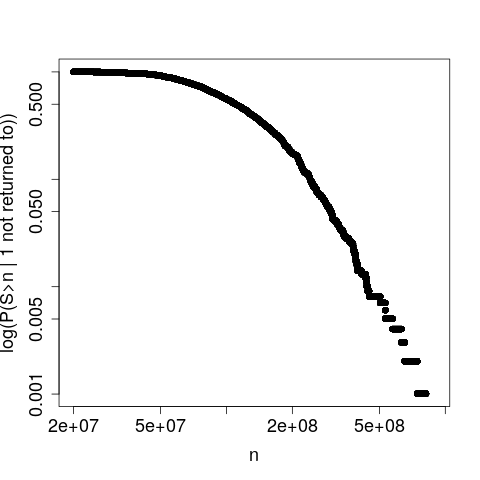}  
\caption{Comparison of the number of steps in the 
$p_{1}$ simulation, using the logarithm of the survival function. 
Left: when state $1$ was returned to, right:
state $1$ was not returned to, i.e. $H>p_{1}^{-1}=10000$. 
}\label{figNumSteps}
\end{figure}

\clearpage
\begin{figure}[!ht]
\centering
\includegraphics[width=0.3\textwidth]{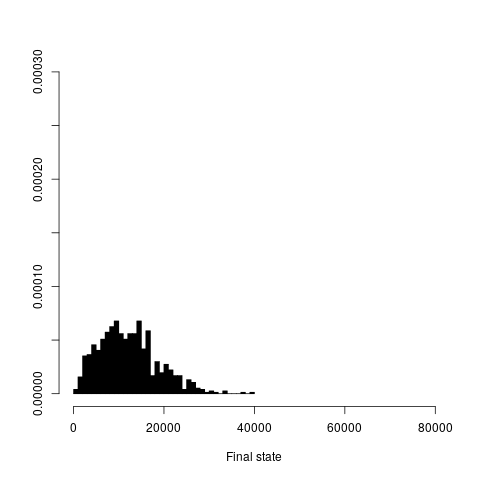}
\includegraphics[width=0.3\textwidth]{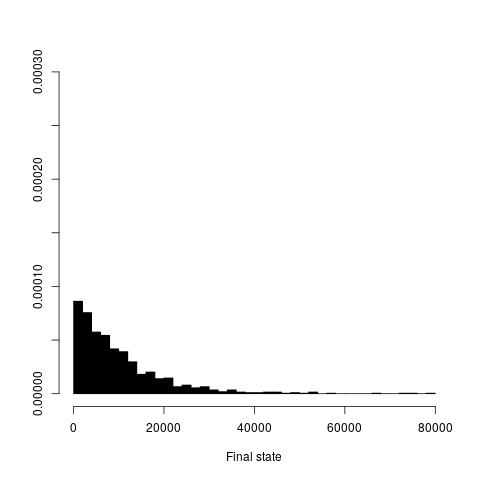}  
\includegraphics[width=0.3\textwidth]{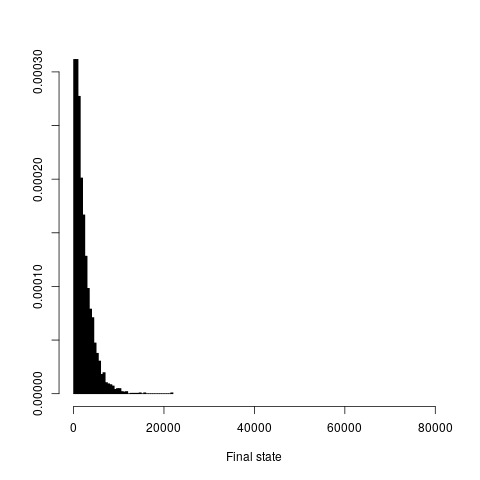}  
\caption{The final state reached when state $1$ was not returned to in the direct
simulation. Left: ``direct'' simulation, centre: $p_{1}$ simulation, right: $p_{2}$ simulation.
In the $p_{1}$ and $p_{2}$ simulations case, in these instances $H$ was obtained by sampling from
\lstinline{poweRlaw::rplcon()}. 
}\label{figFinalStates}
\end{figure}

\section{Acknowledgments}
KB's research is supported by  the Swedish Research Council's (Vetenskapsr\aa det) grant no. $2017$--$04951$.
KB would like to thank the Editor and an anonymous Reviewer for careful reading of the manuscript and
comments that greatly improved it.
KB is grateful to Serik Sagitov for encouragement and suggestions to study general passage times between $m$ and
$k$ in the considered here birth--death process.

\newpage
\textbf{Appendix: R code for simulating $H$ used in Section \ref{secSimulStudy}}
\\~\\
{\small
\begin{lstlisting}
f_dosim_Exact<-function(i,maxstep){
## Simulation of H directly from the Markov model
    H<-0;waitT<-0;X<-2
    b_1<-FALSE;n<-1
    while((X!=1)&&(n<maxstep)){
	n<-n+1
	H<-H+rexp(1,rate=2*X) ## for lambda=1 
	X<-X+sample(c(-1,1),1) ## for lambda=1 
    }
    if (X==1){b_1<-TRUE}
    T1<-rexp(1)
    c(T=T1+H,T1=T1,H=H,n=n,X=X,b_1=b_1)
}

\end{lstlisting}
\begin{lstlisting}
f_dosim_PLcon<-function(i,p){
## simulation of H using the hybrid algorithm
    Ttail<-1/p
    T1<-rexp(1)
    H<-0;X<-2
    b_1<-FALSE;n<-1;sim_H<-0
    while((X!=1)&&(sim_H<=Ttail)){
	n<-n+1
	moveT<-rexp(1,rate=2*X) ## for lambda=1
	sim_H<-sim_H+moveT
	if (sim_H<=Ttail){
	    H<-H+moveT
	    X<-X+sample(c(-1,1),1)  ## for lambda=1
	}
    }
    if ((X==1)&&(sim_H<=Ttail)){b_1<-TRUE}
    else{
	H<-poweRlaw::rplcon(1,xmin=Ttail,alpha=2)
    }
    c(T=H+T1,T1=T1,H=H,n=n,X=X,b_1=b_1,Ttail=Ttail,sim_H=sim_H,p=p)
}

\end{lstlisting}
}

\end{document}